\documentclass[review]{elsarticle}

\usepackage{lineno,hyperref}
\usepackage{amssymb}
\usepackage{color}
\usepackage{amsmath}
\usepackage{tikz}
\usepackage{amsfonts}
\usepackage{mathrsfs}
\usepackage{amsthm}

\usepackage{bbding}
\usepackage{mathrsfs}
\usepackage{amsmath, amsfonts, amssymb, mathrsfs, txfonts}
\usepackage{graphicx, subfigure}
\usepackage{wasysym}
\usepackage{color}
\usepackage{bm}
\usepackage{enumerate}

\usepackage{pifont}
\usepackage{txfonts}
\usepackage{amsmath}
\usepackage{graphicx}
\usepackage{amsfonts}
\usepackage{amssymb}
\usepackage{mathrsfs,psfrag,eepic,epsfig}
\usepackage{epstopdf}

\biboptions{numbers,sort&compress}
\newtheorem{thm}{Theorem}[section]
\newtheorem{cor}[thm]{Corollary}

\newtheorem{prop}[thm]{Proposition}
\theoremstyle{definition}

\modulolinenumbers[5]

\journal{Journal of \LaTeX\ Templates}

\makeatletter \@addtoreset{equation}{section}
\renewcommand{\theequation}{\arabic{section}.\arabic{equation}}

\begin{document}

\begin{frontmatter}

\title{Riemann-Hilbert approach and $N$-soliton solutions for the three-component coupled Hirota equations}
\tnotetext[mytitlenote]{Project supported by the Fundamental Research Fund for the Central Universities under the grant No. 2019ZDPY07.\\
\hspace*{3ex}$^{*}$Corresponding author.\\
\hspace*{3ex}\emph{E-mail addresses}: xwu@cumt.edu.cn (X. Wu), sftian@cumt.edu.cn,
shoufu2006@126.com (S. F. Tian),
and jinjieyang@cumt.edu.cn (J.J. Yang)}

\author{Xin Wu, Shou-Fu Tian$^{*}$, Jin-Jie Yang}
\address{
School of Mathematics and Institute of Mathematical Physics, China University of Mining and Technology, Xuzhou 221116, People's Republic of China
}

\begin{abstract} In this work, we consider an integrable three-component coupled Hirota (tcCH) equations in detail via the Riemann-Hilbert (RH) approach. We present some properties of the spectral problems of the tcCH equations with $4\times4$ the Lax pair. Moreover, the RH problem of the equations is established via analyzing the analyticity of the spectrum problem. By studying the symmetry of the spectral problem, we get the spatiotemporal evolution of scattering data. Finally, the $N$-soliton solution is derived by solving the RH problem with reflectionless case.
According to the resulting $N$-soliton  solution, the influences of each parameters on collision dynamic behaviors between solitons are discussed, and the method of how to control the  interactions are suggested by some graphic analysis.  In addition, some new phenomenon for soliton collision is presented including localized structures and dynamic behaviors of one- and two- soliton solutions, which can help enrich the nonlinear dynamics of the $N$-component nonlinear Schr\"{o}dinger type equations.
\end{abstract}


\end{frontmatter}


\section{Introduction}
Hirota equation describing plane selffocusing and one-dimensional self-modulation of waves \cite{Hirota-1973} is further developed based on nonlinear Schr\"{o}dinger equation which is a kind of very important model in hydrodynamics, optical fiber transmission, etc. The Hirota equation reads \cite{Ankiewicz-2010}
\begin{align}
iq_{t}+\frac{1}{2}q_{tt}+|q|^{2}q-i\alpha_{3}q_{ttt}-i6\alpha_{3}|q|^{2}q_{t}=0,
\end{align}
in dimensionless form. The Hirota equation reduces the nonlinear Schr\"{o}dinger equation when $\alpha_{3}=0$.
After that, the coupled Hirota equations are proposed and considered to describer the pulse propagation in a coupled fiber with higher-order dispersion and self-steepening. The equations read
\begin{align}
&\left\{\begin{aligned}
&iu_{t}+\frac{1}{2}u_{xx}+\left(|u|^{2}+|v|^{2}\right)u+i\epsilon\left[u_{xxx}+(6|u|^{2}+3|v|^{2})u_{x}+3uv^{*}v_{x}\right]=0,\\
&iv_{t}+\frac{1}{2}v_{xx}+\left(|u|^{2}+|v|^{2}\right)v+i\epsilon\left[v_{xxx}+(6|u|^{2}+3|v|^{2})v_{x}+3vu^{*}u_{x}\right]=0.
\end{aligned}\right.
\end{align}
Its Darboux transformation \cite{Tasgal-1992} 
has been studied, which yields  a lot of rich and good properties for  the coupled Hirota equations.
But the coupled Hirota equations can not be used to solve the compatibility problems of the wavelength division multiplexing and linear system \cite{Bindu-2001}, then a new form of Hirota system is proposed which is the three-component coupled Hirota equations (tcCH)
\begin{align} \label{Q1}
&\left\{ \begin{aligned}
iq_{1t}&+\frac{1}{2}q_{1xx}+\left(|q_{1}|^{2}+|q_{2}|^{2}+|q_{3}|^{2}\right)|q_{1}| \\
&+i\epsilon\left[q_{1xxx}+3\left(2|q_{1}|^{2}+|q_{2}|^{2}+|q_{3}|^{2}\right)
q_{1x}+3q_{1}\left(q_{2}^{\ast}q_{2x}+q_{3}^{\ast}q_{3x}\right)\right]=0,
\\
iq_{2t}&+\frac{1}{2}q_{2xx}+\left(|q_{1}|^{2}+|q_{2}|^{2}+|q_{3}|^{2}\right)|q_{2}|\\
&+i\epsilon\left[q_{2xxx}+3\left(2|q_{2}|^{2}+|q_{3}|^{2}+|q_{1}|^{2}\right)
q_{2x}+3q_{2}\left(q_{3}^{\ast}q_{3x}+q_{1}^{\ast}q_{1x}\right)\right]=0,
\\
iq_{3t}&+\frac{1}{2}q_{3xx}+\left(|q_{1}|^{2}+|q_{2}|^{2}+|q_{3}|^{2}\right)|q_{3}|\\
&+i\epsilon\left[q_{3xxx}+3\left(2|q_{3}|^{2}+|q_{1}|^{2}+|q_{2}|^{2}\right)
q_{3x}+3q_{3}\left(q_{1}^{\ast}q_{1x}+q_{2}^{\ast}q_{2x}\right)\right]=0,
     \end{aligned}  \right.
\end{align}
where $q_{1}(x)$, $q_{2}(x)$ and $q_{3}(x)$ are complex envelops, $q^{*}_{i}(i=1,2,3)$ is the complex conjugate of $q_{i}$, and $\epsilon$ denotes the strength of high-order effects which is a small dimensionless real parameter. The Lax pair of the tcCH equations has been derived by  Bindu \cite{Bindu-2001}, as well as rogue wave and breather wave solutions of the tcCH equations \eqref{Q1} have been obtained in \cite{T. Xu-2017}.

The inverse scattering transformation is a powerful analytical tool to solve integrable systems, which plays an indispensable role in the field of nonlinear sciences. Riemann-Hilbert (RH)  approach is developed by Zakharov et al \cite{Zakharov-1984} based on the theory of inverse scattering transformation which is applied to the field of integrable systems. In recent years, RH approach has been used to study a lot of works in solving integrable models \cite{RH-1}-\cite{Tian-NZyang}. The main purpose of this work is to find more abundant  $N$-soliton solutions of the tcCH equations \eqref{Q1}, and revealing the propagation behavior of the solutions via the RH approach.

The structure of this work is given as follows. In section 2, we analyze the spectrum problem of the tcCH equations in detail and get the analytical propertes of Jost functions. In section 3, the RH problem is established based on the previous conclusions. Moreover, we study the symmetry of scattering matrix and the time-spatial revolutions of the scattering data. In section 4, we can derive the $N$-soliton solution of the tcCH equations \eqref{Q1} via solving the resulting RH problem. In addition, the propagation behavior of soliton solution is analyzed by taking the single-soliton  and two-soliton solutions for examples. Some conclusions and discussions are presented in the final section.

\section{Direct Scattering Transform}

In this section, we shall investigate the RH problem of the Eq.\eqref{Q1} via the direct scattering transform. The Lax pair of the tcCH equations reads
\begin{align}\label{Q3}
&\left\{ \begin{aligned}
&\Phi_{x}=U\Phi, U=\lambda U_{0}+U_{1},\\
&\Phi_{t}=V\Phi, V=\lambda^{3}V_{0}+\lambda^{2}V_{1}+\lambda V_{2}+V_{3},
     \end{aligned}  \right.
\end{align}
where \begin{align*}
&U_{0}=\frac{1}{12\epsilon}\left(\begin{array}{cccc}
    -2i  & 0 &  0 & 0 \\
    0 & i & 0 & 0  \\
    0 & 0 & i & 0 \\
    0 & 0 & 0 & i\\
  \end{array}\right),
  U_{1}=\left(\begin{array}{cccc}
    0 & -q_{1} & -q_{2} & -q_{3} \\
    q_{1}^{\ast} & 0 & 0 & 0  \\
    q_{2}^{\ast} & 0 & 0 & 0 \\
    q_{3}^{\ast} & 0 & 0 & 0\\
 \end{array}\right),
  \end{align*}
\begin{align*}
V_{0}=\frac{1}{16\epsilon}U_{0},
 V_{1}=\frac{1}{8\epsilon}U_{0}+\frac{1}{16\epsilon}U_{1},
 \end{align*}
\begin{align*}
V_{2}=\frac{1}{4}\left(\begin{array}{cccc}
    ie  & -\frac{q_{1}}{2\epsilon}-iq_{1x} & -\frac{q_{2}}{2\epsilon}-iq_{2x} & -\frac{q_{3}}{2\epsilon}-iq_{3x} \\
    \frac{q_{1}^{\ast}}{2\epsilon}-i{q^{\ast}_{1x}} & -i|q_{1}|^{2} & -i{q_{1}}^{\ast}q_{2} & -i{q_{1}}^{\ast}q_{3}  \\
    \frac{q_{2}^{\ast}}{2\epsilon}-i{q^{\ast}_{2x}} & -i{q_{2}}^{\ast}q_{1} & -i|q_{2}|^{2} & -i{q_{2}}^{\ast}q_{3} \\
    \frac{q_{3}^{\ast}}{2\epsilon}-i{q^{\ast}_{3x}} & -i{q_{3}}^{\ast}q_{1} & -i{q_{3}}^{\ast}q_{2} & -i|q_{3}|^{2}\\
  \end{array}\right),
 \end{align*}
\begin{align*}
  V_{3}=\left(\begin{array}{cccc}
    \epsilon\left(e_{1}+e_{2}+e_{3}\right)+\frac{i}{2}e & \epsilon e_{4}-\frac{i}{2}q_{1x} &
    \epsilon e_{5}-\frac{i}{2}q_{2x} & \epsilon e_{6}-\frac{i}{2}q_{3x} \\
    -\epsilon e^{\ast}_{4}-\frac{i}{2}q^{\ast}_{1x} & -\epsilon e_{4}-\frac{i}{2}|q_{1}|^{2} &
     \epsilon e_{7}-\frac{i}{2}q^{\ast}_{1}q_{2} & \epsilon e_{8}-\frac{i}{2}q^{\ast}_{1}q_{3}  \\
    -\epsilon e^{\ast}_{5}-\frac{i}{2}q^{\ast}_{2x} & -\epsilon e^{\ast}_{7}-\frac{i}{2}q^{\ast}_{2}q_{1} &
    -\epsilon e_{2}-\frac{i}{2}|q_{2}|^{2} & \epsilon e_{9}-\frac{i}{2}q^{\ast}_{2}q_{3} \\
    -\epsilon e^{\ast}_{6}-\frac{i}{2}q^{\ast}_{3x} & -\epsilon e^{\ast}_{8}-\frac{i}{2}q^{\ast}_{3}q_{1} &
     -\epsilon e^{\ast}_{9}-\frac{i}{2}q^{\ast}_{3}q_{2} & -\epsilon e_{3}-\frac{i}{2}|q_{3}|^{2}\\
 \end{array}\right),
  \end{align*}
with
\begin{align*}
&e=|q_{1}|^{2}+|q_{2}|^{2}+|q_{3}|^{2},
e_{1}=q_{1}q^{\ast}_{1x}-q_{1x}q^{\ast}_{1},
e_{2}=q_{2}q^{\ast}_{2x}-q_{2x}q^{\ast}_{2},\\
&e_{3}=q_{3}q^{\ast}_{3x}-q_{3x}q^{\ast}_{3},
e_{4}=q_{1xx}+2eq_{1},
e_{5}=q_{2xx}+2eq_{2},
e_{6}=q_{3xx}+2eq_{3},\\
&e_{7}=q^{\ast}_{1}q_{2x}-q^{\ast}_{1x}q_{2},
e_{8}=q^{\ast}_{1}q_{3x}-q^{\ast}_{1x}q_{3},
e_{9}=q^{\ast}_{2}q_{3x}-q^{\ast}_{2x}q_{3},
\end{align*}
where $\Phi=\Phi(x,t;\lambda)$ is column vector function, and $\lambda$ is the complex spectral parameter. Eq. \eqref{Q1} satisfies zero curvature equation $U_{t}-V_{x}+[U,V]=0$, which is the compatibility condition of the lax pair \eqref{Q3}.

For the convenience of calculations, we introduce a new Jost function $J=J(x,t;\lambda)$
\begin{align} \label{Q4}
\Phi=Je^{\frac{1}{12\epsilon}+\left(\frac{1}{192\epsilon^{2}}i\lambda^{3}+\frac{1}{96\epsilon^{2}}i\lambda^{2}\right)\sigma t},
\end{align}
where $\sigma=diag(-2,1,1,1)$. According to Eq. \eqref{Q4}, the lax pairs can be converted to
\begin{align}\label{Q5}
\left\{ \begin{aligned}
&J_{x}=i\lambda c_{1}[\sigma,J]+QJ,\\
&J_{t}=i\left(c_{2}\lambda^{3}+c_{3}\lambda^{2}\right)[\sigma,J]+\tilde{V}J,
     \end{aligned} \right.
\end{align}
where
$$Q=U_{1},\tilde{V}=\lambda V_{2}+V_{3},c_{1}=\frac{1}{12\epsilon},c_{2}=\frac{1}{192\epsilon^{2}},c_{3}=\frac{1}{96\epsilon^{2}},$$
and $[\sigma,J]$ implies that $[\sigma,J]=\sigma J-J\sigma$ .

In the study of the symmetry of matrix $Q$, we only consider the first expression of Eq.\eqref{Q5}. The another expression is useful in the process of inverse scattering, so it is temporarily omitted.

Next let us consider two solutions $J_{\pm}=J_{\pm}(x,\lambda)$ of the first expression of Eq.\eqref{Q5} for $\lambda\in\mathbb{R}$

\begin{align}\label{Q6}
\left\{\begin{aligned}
J_{-}=\left([J_{-}]_{1},[J_{-}]_{2},[J_{-}]_{3},[J_{-}]_{4}\right),\\
J_{+}=\left([J_{+}]_{1},[J_{+}]_{2},[J_{+}]_{3},[J_{+}]_{4}\right),
\end{aligned}\right.
\end{align}
with the asymptotic conditions
\begin{align}\label{Q7}
\begin{split}
&J_{-}\rightarrow\mathbb{I}, \qquad x\rightarrow -\infty,\\
&J_{+}\rightarrow\mathbb{I}, \qquad x\rightarrow +\infty,
\end{split}
\end{align}
where $\mathbb{I}$ represents a $4\times4$ indentity matrix.

Next, we study the analytic properties of $J_{\pm}(x,\lambda)$ and give the following proposition.

\begin{prop}
\begin{align}
[J_{-}]_{1},[J_{+}]_{2},[J_{+}]_{3},[J_{+}]_{4}
 \end{align}
 allow analytic extensions to the upper half $\lambda$-plane $\mathbb{C}^{+}$;
\begin{align}
[J_{+}]_{1},[J_{-}]_{2},[J_{-}]_{3},[J_{-}]_{4}
 \end{align}
 allow analytic extensions to the lower half $\lambda$-plane $\mathbb{C}^{-}$.
\end{prop}

\begin{proof}
Two solutions are completely determined by Volterra integrable equations
\begin{align}
&J_{-}(x,\lambda)=\mathbb{I}+\int_{-\infty}^{x}
e^{ic_{1}\lambda\sigma(x-y)}Q(y)J_{-}(y,\lambda)e^{-ic_{1}\lambda\sigma(x-y)}dy, \label{Q8} \\
&J_{+}(x,\lambda)=\mathbb{I}-\int_{x}^{+\infty}
e^{ic_{1}\lambda\sigma(x-y)}Q(y)J_{+}(y,\lambda)e^{-ic_{1}\lambda\sigma(x-y)}dy.\label{Q9}
\end{align}
According to \cite{Biondini-2014}, the analyticity of $J_{\pm}$ is equivalent to that of integrand function
\begin{align*}
e^{ic_{1}\lambda\hat{\sigma}(x-y)}Q(y)=\left(\begin{array}{cccc}
    0 & -q_{1}e^{-3ic_{1}(x-y)} & -q_{2}e^{-3ic_{1}(x-y)} & -q_{3}e^{-3ic_{1}(x-y)} \\
    q_{1}^{\ast}e^{3ic_{1}(x-y)} & 0 & 0 & 0  \\
    q_{2}^{\ast}e^{3ic_{1}(x-y)} & 0 & 0 & 0 \\
    q_{3}^{\ast}e^{3ic_{1}(x-y)} & 0 & 0 & 0\\
 \end{array}\right),
\end{align*}
where $e^{\hat{\sigma}}A=e^{\sigma}Ae^{-\sigma}$.

To find the analytic area of each column, we just consider $Re[3ic_{1}¦Ë(x-y)] < 0$ and $Re[-3ic_{1}¦Ë(x-y)] < 0$. Because of $x-y>0$ in Eq.\eqref{Q8}, we can easily see that $[J_{-}]_{1}$ allows analytic extensions to the upper half $\lambda$-plane $\mathbb{C}^{+}$; $[J_{-}]_{2},[J_{-}]_{3}$ as well as $[J_{-}]_{4}$ allow analytic extensions to the lower half $\lambda$-plane $\mathbb{C}^{-}$. Similarly, $[J_{+}]_{1}$ allows analytic extensions to the lower half $\lambda$-plane $\mathbb{C}^{+}$; $[J_{+}]_{2},[J_{+}]_{3}$ as well as $[J_{+}]_{4}$ allow analytic extensions to the upper half $\lambda$-plane $\mathbb{C}^{-}$.
\end{proof}

To prove the following conclusion, we first introduce a theorem.
\begin{thm}
(Abel's indentity) Suppose $A(x)\in\mathbb{C}^{n\times n}$,
\begin{align*}
Y_{x}=A(x)Y,
\end{align*}
then we obtain
\begin{align*}
(\det Y)_{x}=(trA)\det Y.
\end{align*}
Furthermore, we have
\begin{align*}
\det Y(x)=\det Y(x_{0})e^{\int_{x_{0}}^{x}trA(t)dt}.
\end{align*}
\end{thm}

\begin{proof}
Introduce
\begin{align*}
A=\left(\begin{array}{cccc}
a_{11} & a_{12} & \cdots & a_{1n}\\
a_{21} & a_{22} & \cdots & a_{2n}\\
\vdots & \vdots & \ddots & \vdots\\
a_{n1} & a_{n2} & \cdots &a_{nn}
\end{array}
\right),
Y=\left(\begin{array}{c}
Y_{1}\\
Y_{2}\\
\vdots\\
Y_{n}\\
\end{array}
\right),
\end{align*}
where $Y_{j}$ is the $j$th row of matrix Y. From $Y_{x}=A(x)Y$, we have
\begin{align*}
Y_{j,x}=a_{j1}Y_{1}+a_{j2}Y_{2}+\cdots+a_{jn}Y_{n},j=1,2,\cdots,n.
\end{align*}
Then, we know that
\begin{align*}
\left(\det Y\right)_{x}=\sum_{j=1}^{n}\det\left(\begin{array}{c}
Y_{1}\\
\vdots\\
Y_{j,x}\\
\vdots\\
Y_{n}\\
\end{array}\right)
=\sum_{j=1}^{n}\det\left(\begin{array}{c}
Y_{1}\\
\vdots\\
a_{j1}Y_{1}+\cdots+a_{jj}Y_{j}+\cdots+a_{jn}Y_{n}\\
\vdots\\
Y_{n}\\
\end{array}\right),
\end{align*}
which implies
\begin{align*}
\left(\det Y\right)_{x}=\sum_{j=1}^{n}a_{jj}\det Y=(trA)\det Y.
\end{align*}
Next, we can get $\det Y=C_{0}e^{\int_{x_{0}}^{x}trA(t)dt}$ through integral. Then making $x=x_{0}$, we can obtain $C_{0}=\det Y(x_{0})$. Thus
\begin{align*}
\det Y(x)=\det Y(x_{0})e^{\int_{x_{0}}^{x}trA(t)dt}.
\end{align*}
Now we finish this proof.
\end{proof}

\begin{cor}
Since $tr(Q)=0$, we know that
\begin{align}\label{Q10}
\det J_{\pm}=1,\lambda\in\mathbb{R}.
\end{align}
\end{cor}

Introducing the notation $E=e^{ic_{1}\lambda\sigma x}$, one can know that $J_{-}E$ and $J_{+}E$ are matrix solutions of the $x$-part of Eq.\eqref{Q5}, so they are linear dependent. There is a $4\times4$ scattering matrix $S(\lambda)=(s_{kj})_{4\times4}$ to make these two solutions satisfy
\begin{align}\label{Q11}
J_{-}E=J_{+}S(\lambda),\lambda\in\mathbb{R}.
\end{align}

Eq.\eqref{Q10} and Eq.\eqref{Q11} imply that
\begin{align}\label{Q12}
\det S(\lambda)=1,\lambda\in\mathbb{R}.
\end{align}

In order to establish the RH problem, we need to consider the inverse matrix of $J_{\pm}$ and block it by rows. Marking it as
\begin{align}\label{Q13}
J^{-1}_{\pm}=\left(\begin{array}{c}
\left[J^{-1}_{\pm}\right]^{1}\\
\left[J^{-1}_{\pm}\right]^{2}\\
\left[J^{-1}_{\pm}\right]^{3}\\
\left[J^{-1}_{\pm}\right]^{4}
\end{array}\right),
\end{align}
and it is not difficult to check that $J^{-1}_{\pm}$ is determined by the equation of $K$
\begin{align}\label{Q14}
K_{x}=c_{1}i\lambda[\sigma,K]-KQ.
\end{align}
From Eq.\eqref{Q14}, one can know that $\left[J^{-1}_{-}\right]^{1},\left[J^{-1}_{+}\right]^{2}, \left[J^{-1}_{+}\right]^{3},\left[J^{-1}_{+}\right]^{4}$ are analytically extendible to the lower half $\lambda$-plane $\mathbb{C}^{-}$, when $\left[J^{-1}_{+}\right]^{1},\left[J^{-1}_{-}\right]^{2}, \left[J^{-1}_{-}\right]^{3},\left[J^{-1}_{-}\right]^{4}$ are analytically extendible to the upper half $\lambda$-plane $\mathbb{C}^{+}$.

Further,
\begin{align}\label{Q15}
E^{-1}J^{-1}_{-}=R(\lambda)E^{-1}J^{-1}_{+},\qquad\lambda\rightarrow\infty,
\end{align}
can be obtained from Eq.\eqref{Q11}. Here $R(\lambda)=(r_{ij})_{4\times4}=S^{-1}(\lambda)$, is called the scattering matrix.

\begin{prop}
The analytic properties of $s_{ij},r_{ij}$ are given as:\\
$s_{11}$  allows analytic extensions to the upper half $\lambda$-plane $\mathbb{C}^{+}$; $s_{22}$, $s_{23}$, $s_{24}$, $s_{32}$, $s_{33}$, $s_{34}$, $s_{42}$, $s_{43}$ and $s_{44}$ allow analytic extensions to the lower half $\lambda$-plane $\mathbb{C}^{+}$; $s_{12}$, $s_{13}$, $s_{14}$, $s_{21}$, $s_{31}$, $s_{41}$ cannot be extended off the real $\lambda$-axis. $r_{11}$  allows analytic extensions to the upper half $\lambda$-plane $\mathbb{C}^{+}$; $r_{22}$, $r_{23}$, $r_{24}$, $r_{32}$, $r_{33}$, $r_{34}$, $r_{42}$, $r_{43}$ and $r_{44}$ allow analytic extensions to the lower half $\lambda$-plane $\mathbb{C}^{+}$; $r_{12}$, $r_{13}$, $r_{14}$, $r_{21}$, $r_{31}$, $r_{41}$ cannot be extended off the real $\lambda$-axis.
\end{prop}

\begin{proof}
Resorting to Eq.\eqref{Q11} and Eq.\eqref{Q15}, one knows that
$$E^{-1}J^{-1}_{+}J_{-}E=S(\lambda),$$
hence
$$
S(\lambda)=E^{-1}\left(\begin{array}{cccc}
\left[J^{-1}_{+}\right]^{1}\left[J_{-}\right]_{1} & \left[J^{-1}_{+}\right]^{1}\left[J_{-}\right]_{2} &
 \left[J^{-1}_{+}\right]^{1}\left[J_{-}\right]_{3} & \left[J^{-1}_{+}\right]^{1}\left[J_{-}\right]_{4}\\
\left[J^{-1}_{+}\right]^{2}\left[J_{-}\right]_{1} & \left[J^{-1}_{+}\right]^{2}\left[J_{-}\right]_{2} &
 \left[J^{-1}_{+}\right]^{2}\left[J_{-}\right]_{3} & \left[J^{-1}_{+}\right]^{2}\left[J_{-}\right]_{4}\\
\left[J^{-1}_{+}\right]^{3}\left[J_{-}\right]_{1} & \left[J^{-1}_{+}\right]^{3}\left[J_{-}\right]_{2} &
 \left[J^{-1}_{+}\right]^{3}\left[J_{-}\right]_{3} & \left[J^{-1}_{+}\right]^{3}\left[J_{-}\right]_{4}\\
\left[J^{-1}_{+}\right]^{4}\left[J_{-}\right]_{1} & \left[J^{-1}_{+}\right]^{4}\left[J_{-}\right]_{2} &
 \left[J^{-1}_{+}\right]^{4}\left[J_{-}\right]_{3} & \left[J^{-1}_{+}\right]^{4}\left[J_{-}\right]_{4}
\end{array}\right)E
.$$
According to the analytic properties of $[J^{-1}_{+}]$ and $[J_{-}]$, we can proof the proposition. Similarly, the analytic properties of $r_{ij}$ can be obtained.
\end{proof}

\section{Riemann-Hilbert problem}

To construct the RH problem, we need to seek the two analytical functions in different regions $\mathbb{C}^{\pm}$. Take
\begin{align}\label{Q16}
P_{1}=\left([J_{-}]_{1},[J_{+}]_{2},[J_{+}]_{3},[J_{+}]_{4}\right),
\end{align}
which is analytical in $\mathbb{C}^{+}$. In addition, the following asymptotic behavior of $P_{1}$ can be determined by
\begin{align}\label{Q17}
P_{1}\rightarrow\mathbb{I},\qquad\lambda\rightarrow\infty.
\end{align}
Similarly, constructing a matrix function $P_{2}=P_{2}(x,\lambda)$ is analytic for $\lambda$ in $\mathbb{C}^{-}$
\begin{align}\label{Q18}
P_{2}=\left(\begin{array}{c}
\left[J^{-1}_{-}\right]^{1}\\
\left[J^{-1}_{+}\right]^{2}\\
\left[J^{-1}_{+}\right]^{3}\\
\left[J^{-1}_{+}\right]^{4}
\end{array}\right).
\end{align}
Furthermore, the following asymptotic behavior of $P_{2}$ can be determined by
\begin{align}\label{Q19}
P_{2}\rightarrow\mathbb{I},\qquad\lambda\rightarrow\infty.
\end{align}

Based on the about results, the RH problem of tcCH equations can be formed.
\begin{thm}
Denoting the limit of $P_{1}$ on the left side of the real $\lambda$-axis is $P^{+}$, and the limit of $P_{2}$ on the right side of the real $\lambda$-axis is $P^{-}$. We can obtain
\begin{align}
&P^{\pm} \text{ is analytic in } \mathbb{C}^{\pm};\\
&P^{-}(x,\lambda)P^{+}(x,\lambda)=G(x,\lambda), \qquad \lambda\in\mathbb{R};\label{Q20}\\
&P^{\pm}(x,\lambda)\rightarrow \mathbb{I}, \quad\quad \lambda \rightarrow \infty,
\end{align}
where
\begin{align*}
G\left(x,\lambda\right)=\left(\begin{array}{cccc}
1 & r_{12}e^{-3ic_{1}\lambda x} & r_{13}e^{-3ic_{1}\lambda x} & r_{14}e^{-3ic_{1}\lambda x}\\
s_{21}e^{3ic_{1}\lambda x} & 1 & 0 & 0\\
s_{31}e^{3ic_{1}\lambda x} & 0 & 1 & 0\\
s_{41}e^{3ic_{1}\lambda x} & 0 & 0 & 1
\end{array}\right).
\end{align*}
\end{thm}

To sum up the above analysis, the following conclusions are drawn.
\begin{align}\label{Q21}
\det P_{1}(\lambda)=s_{11}(\lambda),\qquad\lambda\in\mathbb{C}^{+},\\
\det P_{2}(\lambda)=r_{11}(\lambda),\qquad\lambda\in\mathbb{C}^{-}.
\end{align}

Noticing $Q^{\dagger}=-Q$, here $\dagger$ means the Hermitian of a matrix, the following conclusions can be drawn.
\begin{prop}
\begin{align}\label{Q22}
J^{\dagger}_{\pm}(\lambda^{*})=J^{-1}_{\pm}(\lambda).
\end{align}
\end{prop}

\begin{proof}
According to the first equation of Eq.\eqref{Q5}, we can know that
\begin{align*}
J_{\pm,x}(\lambda^{*})=ic_{1}\lambda^{*}[\sigma,J_{\pm}(\lambda^{*})]+QJ_{\pm}(\lambda^{*}).
\end{align*}
Taking conjugate transposition on both sides of the equation, we can get
\begin{align*}
J^{\dagger}_{\pm,x}(\lambda^{*})
=ic_{1}\lambda\left[\sigma,J^{\dagger}_{\pm}(\lambda^{*})\right]-J^{\dagger}_{\pm}(\lambda^{*})Q.
\end{align*}
The expression $\left(J_{\pm}J^{-1}_{\pm}\right)_{x}=I_{x}=0$ implies that
\begin{align*}
J^{-1}_{\pm,x}
=ic_{1}\lambda\left[\sigma,J^{-1}_{\pm}\right]-J^{-1}_{\pm}Q.
\end{align*}
Owing to $J^{\dagger}_{\pm}\left(\lambda^{*}\right)$ and $J^{-1}_{\pm}\left(\lambda\right)$ also satisfy Eq.\eqref{Q14}, the above proposition is proved.
\end{proof}

\begin{prop}
\begin{align}\label{Q23}
S^{\dagger}(\lambda^{*})=S^{-1}(\lambda).
\end{align}
\end{prop}

\begin{proof}
Resorting to Eq.\eqref{Q11}, one can see that
\begin{align*}
J_{-}=J_{+}ES(\lambda)E^{-1},
\end{align*}
which means
\begin{align*}
J_{-}(\lambda^{*})=J_{+}(\lambda^{*})e^{ic_{1}\lambda^{*}\sigma x}S(\lambda^{*})e^{-ic_{1}\lambda^{*}\sigma x}.
\end{align*}
Taking the conjugate transpose of both sides of the equation yields
\begin{align*}
J^{\dagger}_{-}(\lambda^{*})=e^{ic_{1}\lambda\sigma x}S^{\dagger}(\lambda^{*})e^{-ic_{1}\lambda\sigma x}J^{\dagger}_{+}(\lambda^{*}),
\end{align*}
thus it follows from Eq.\eqref{Q22} that
\begin{align*}
J^{-1}_{-}(\lambda)=e^{ic_{1}\lambda\sigma x}S^{\dagger}(\lambda^{*})e^{-ic_{1}\lambda\sigma x}J^{-1}_{+}(\lambda).
\end{align*}
Combining with
\begin{align*}
J^{-1}_{-}(\lambda)=e^{ic_{1}\lambda\sigma x}S^{-1}(\lambda)e^{-ic_{1}\lambda\sigma x}J^{-1}_{+}(\lambda),
\end{align*}
the proof of proposition is finished.
\end{proof}

Obviously, it is easy to obtain that
\addtocounter{equation}{1}
\begin{align}
&s_{11}(\lambda)=r^{*}_{11}(\lambda^{*}),\qquad\lambda\in\mathbb{C}^{+},\tag{\theequation a}\label{Q24}\\
&s^{*}_{21}(\lambda)=r_{12}(\lambda),\qquad\lambda\in\mathbb{R},\tag{\theequation b}\\
&s^{*}_{31}(\lambda)=r_{13}(\lambda),\qquad\lambda\in\mathbb{R},\tag{\theequation c}\\
&s^{*}_{41}(\lambda)=r_{14}(\lambda),\qquad\lambda\in\mathbb{R}.\tag{\theequation d}
\end{align}

Besides, we get
\begin{align}\label{Q25}
P^{\dag}_{1}(\lambda^{*})=P_{2}(\lambda), \qquad \lambda\in\mathbb{C}^{-}.
\end{align}

Therefore, in terms of Eq.\eqref{Q24} as well as Eq.\eqref{Q25}, we can know that if $\lambda$ is a zero of $\det P_{1}$, $\hat{\lambda}=\lambda^{*}$ is a zero of $\det P_{2}$. Assuming that $\det P_{1}$ has $N$ simple zero $\{\lambda_{j}\}^{N}_{1}$ in $\mathbb{C}^{+}$, thus $\det P_{2}$ also has $N$ simple zero $\{\lambda^{*}_{j}\}^{N}_{1}$, which are all in $\mathbb{C}^{-}$. Here, these zeros and the nonzero vectors, which is $w_{j}$ and $w^{*}_{j}$ respectively, constitute the full set of the scattering data, such that
\begin{align}
P_{1}(\lambda_{j})w_{j}=0,\label{Q26}\\
\hat{w}_{j}P_{2}(\hat{\lambda}_{j})=0.
\end{align}
Then one can obtain the following relation
\begin{align}\label{Q27}
w^{\dagger}_{j}=\hat{w}_{j},\qquad 1\leq j\leq N.
\end{align}

\begin{prop}
The time-spatial revolutions of $w_{j}$ and $\hat{w}_{j}$ are listed below:
\begin{align}
&w_{j}=e^{i(c_{1}\lambda x+c_{2}\lambda^{3}t+c_{3}\lambda^{2}t)\sigma}w_{j,0},\qquad 1\leq j\leq N,\label{Q28}\\
&\hat{w}_{j}=w^{\dagger}_{j,0}e^{-i(c_{1}\lambda^{*} x+c_{2}(\lambda^{*})^{3}t+c_{3}(\lambda^{*})^{2}t)\sigma},\qquad 1\leq j\leq N,\label{Q29}
\end{align}
where $w_{j,0}$ is a complex constant vector.
\end{prop}

\section{Multi-soliton Solutions}

For the RH problem Eq.\eqref{Q20}, the solutions (for details, please refer to \cite{yjk-2010}) are given as follows:
\begin{align}
P_{1}(\lambda)=\mathbb{I}-\sum^{N}_{k=1}\sum^{N}_{j=1}\frac{w_{k}\hat{w}_{j}\left(M^{-1}_{kj}\right)}{\lambda-\hat{\lambda}_{j}},\label{Q30}\\
P_{2}(\lambda)=\mathbb{I}+\sum^{N}_{k=1}\sum^{N}_{j=1}\frac{w_{k}\hat{w}_{j}\left(M^{-1}_{kj}\right)}{\lambda-\hat{\lambda}_{j}},\label{Q31}
\end{align}
where $M$ is a $N\times N$ matrix whose elements are $m_{kj}=\frac{\hat{w}_{k}w_{j}}{\lambda_{j}-\hat{\lambda}_{k}}$.

Taking Laurent series expansion for $P_{1}$ yields
\begin{align}
P_{1}\left(\lambda\right)=\mathbb{I}+\lambda^{-1}P_{1}^{(1)}
+\lambda^{-2}P_{1}^{(2)}+\cdots,~~~~\lambda\rightarrow\infty.
\end{align}
Substituting expansion into the first equation of Eq.\eqref{Q5}, we collect the term $\lambda^{0}$ and  obtain
\begin{align}
Q=-ic_{1}[\sigma,P^{(1)}_{1}],
\end{align}
which implies
\begin{align}\label{Q32}
\left\{ \begin{aligned}
q_{1}=-3ic_{1}\left(P^{(1)}_{1}\right)_{12},\\
q_{2}=-3ic_{1}\left(P^{(1)}_{1}\right)_{13},\\
q_{3}=-3ic_{1}\left(P^{(1)}_{1}\right)_{14}.
\end{aligned}\right.
\end{align}

From Eq.\eqref{Q30}, we can directly calculate
\begin{align}\label{Q33}
P^{(1)}_{1}=-\sum^{N}_{k=1}\sum^{N}_{j=1} w_{k}\hat{w}_{j}\left(M^{-1}\right)_{kj}.
\end{align}

Supposing the nonzero vector satisfies $w_{j,0}=\left(\alpha_{j},\beta_{j},\gamma_{j},1\right)^{T}$, and $\theta_{j}=i(c_{1}\lambda x+c_{2}\lambda^{3}t+c_{3}\lambda^{2}t)$ with $\lambda_{j}=\xi_{j}+i\eta_{j}\left(\xi_{j}\neq0,\eta_{j}>0,1\leq j\leq N\right)$, therefore we obtain
\begin{align*}
w_{j}=e^{\theta_{j}\sigma}w_{j,0}=\left(\begin{array}{cccc}
e^{-2\theta_{j}} & 0 & 0 & 0\\
0 & e^{\theta_{j}} & 0 & 0\\
0 & 0 & e^{\theta_{j}} & 0\\
0 & 0 & 0 & e^{\theta_{j}}
\end{array}\right)
\left(\begin{array}{c}
\alpha_{j}\\
\beta_{j}\\
\gamma_{j}\\
1
\end{array}\right)
=\left(\begin{array}{c}
\alpha_{j}e^{-2\theta_{j}}\\
\beta_{j}e^{\theta_{j}}\\
\gamma_{j}e^{\theta_{j}}\\
e^{\theta_{j}}
\end{array}\right),
\end{align*}

\begin{align*}
\hat{w}_{j}(\lambda^{*})=\left(\begin{array}{cccc}
\alpha^{*}_{j}e^{-2\theta^{*}_{j}}, \beta^{*}_{j}e^{\theta^{*}_{j}}, \gamma^{*}_{j}e^{\theta^{*}_{j}}, e^{\theta^{*}_{j}}
\end{array}\right).
\end{align*}
Then
\begin{align*}
w_{k}\hat{w}_{j}=\left(\begin{array}{cccc}
\alpha_{k}\alpha^{*}_{j}e^{-2\left(\theta_{k}+\theta^{*}_{j}\right)} & \alpha_{k}\beta^{*}_{j}e^{-2\theta_{k}+\theta^{*}_{j}}
& \alpha_{k}\gamma^{*}_{j}e^{-2\theta_{k}+\theta^{*}_{j}} & \alpha_{k}e^{-2\theta_{k}+\theta^{*}_{j}}\\
\beta_{k}\alpha^{*}_{j}e^{\theta_{k}-2\theta^{*}_{j}} & \beta_{k}\beta^{*}_{j}e^{\theta_{k}+\theta^{*}_{j}}
& \beta_{k}\gamma^{*}_{j}e^{\theta_{k}+\theta^{*}_{j}} & \beta_{k}e^{\theta_{k}+\theta^{*}_{j}}\\
\gamma_{k}\alpha^{*}_{j}e^{\theta_{k}-2\theta^{*}_{j}} & \gamma_{k}\beta^{*}_{j}e^{\theta_{k}+\theta^{*}_{j}}
& \gamma_{k}\gamma^{*}_{j}e^{\theta_{k}+\theta^{*}_{j}} & \gamma_{k}e^{\theta_{k}+\theta^{*}_{j}}\\
\alpha^{*}_{j}e^{\theta_{k}-2\theta^{*}_{j}} &  \beta^{*}_{j}e^{\theta_{k}+\theta^{*}_{j}}
& \gamma^{*}_{j}e^{\theta_{k}+\theta^{*}_{j}} & e^{\theta_{k}+\theta^{*}_{j}}
\end{array}\right),
\end{align*}

\begin{align*}
\hat{w}_{k}w_{j}=\alpha^{*}_{k}\alpha_{j}e^{-2(\theta^{*}_{k}+\theta_{j})}
+\beta^{*}_{k}\beta_{j}e^{\theta^{*}_{k}+\theta_{j}}
+\gamma^{*}_{k}\gamma_{j}e^{\theta^{*}_{k}+\theta_{j}}
+e^{\theta^{*}_{k}+\theta_{j}}.
\end{align*}

As a result, $N$-soliton solution to the tcCH equations \eqref{Q1} can be derived by
\begin{align}\label{Q34}
\left\{ \begin{aligned}
&q_{1}=3ic_{1}\sum^{N}_{k=1}\sum^{N}_{j=1}\alpha_{k}\beta^{*}_{j}e^{-2\theta_{k}+\theta^{*}_{j}}\left(M^{-1}\right)_{kj},\\
&q_{2}=3ic_{1}\sum^{N}_{k=1}\sum^{N}_{j=1}\alpha_{k}\gamma^{*}_{j}e^{-2\theta_{k}+\theta^{*}_{j}}\left(M^{-1}\right)_{kj},\\
&q_{3}=3ic_{1}\sum^{N}_{k=1}\sum^{N}_{j=1}\alpha_{k}e^{-2\theta_{k}+\theta^{*}_{j}}\left(M^{-1}\right)_{kj},
\end{aligned}\right.
\end{align}
where $M=\left(m_{kj}\right)_{N\times N}$ is defined by
\begin{align*}
m_{kj}
&=\frac{\alpha^{*}_{k}\alpha_{j}e^{-2(\theta^{*}_{k}+\theta_{j})}
+\beta^{*}_{k}\beta_{j}e^{\theta^{*}_{k}+\theta_{j}}
+\gamma^{*}_{k}\gamma_{j}e^{\theta^{*}_{k}+\theta_{j}}
+e^{\theta^{*}_{k}+\theta_{j}}}{\lambda_{j}-\hat{\lambda}_{k}}\\
&=\frac{\alpha^{*}_{k}\alpha_{j}e^{-2(\theta^{*}_{k}+\theta_{j})}+\left(\beta^{*}_{k}\beta_{j}+\gamma^{*}_{k}\gamma_{j}+1\right)e^{\theta^{*}_{k}+\theta_{j}}}{\lambda_{j}-\hat{\lambda}_{k}}.
\end{align*}

Taking $N=1$ in Eq.\eqref{Q34}, single-soliton solution is listed as follows
\begin{align}\label{Q35}
\left\{ \begin{aligned}
&q_{1}=3ic_{1}\alpha_{1}\beta^{*}_{1}e^{-2\theta_{1}+\theta^{*}_{1}}\left(M^{-1}\right)_{11},\\
&q_{2}=3ic_{1}\alpha_{1}\gamma^{*}_{1}e^{-2\theta_{1}+\theta^{*}_{1}}\left(M^{-1}\right)_{11},\\
&q_{3}=3ic_{1}\alpha_{1}e^{-2\theta_{1}+\theta^{*}_{1}}\left(M^{-1}\right)_{11},
\end{aligned}\right.
\end{align}
where
\begin{align*}
m_{11}=\frac{\alpha^{*}_{1}\alpha_{1}e^{-2(\theta^{*}_{1}+\theta_{1})}+\left(\beta^{*}_{1}\beta_{1}+\gamma^{*}_{1}\gamma_{1}+1\right)e^{\theta^{*}_{1}+\theta_{1}}}{\lambda_{1}-\hat{\lambda}_{1}}.
\end{align*}

In what follows, let us discuss the case when $N=1$. Figs. 1-3 show the single-soliton solutions by choosing the appropriate parameters. We find that different parameters have different effects on the solutions, such as the smaller $\epsilon$ yields the higher the peak value, and  the larger $\epsilon$ yields the lower the peak value, etc.

\noindent
{\rotatebox{0}{\includegraphics[width=3.5cm,height=3.5cm,angle=0]{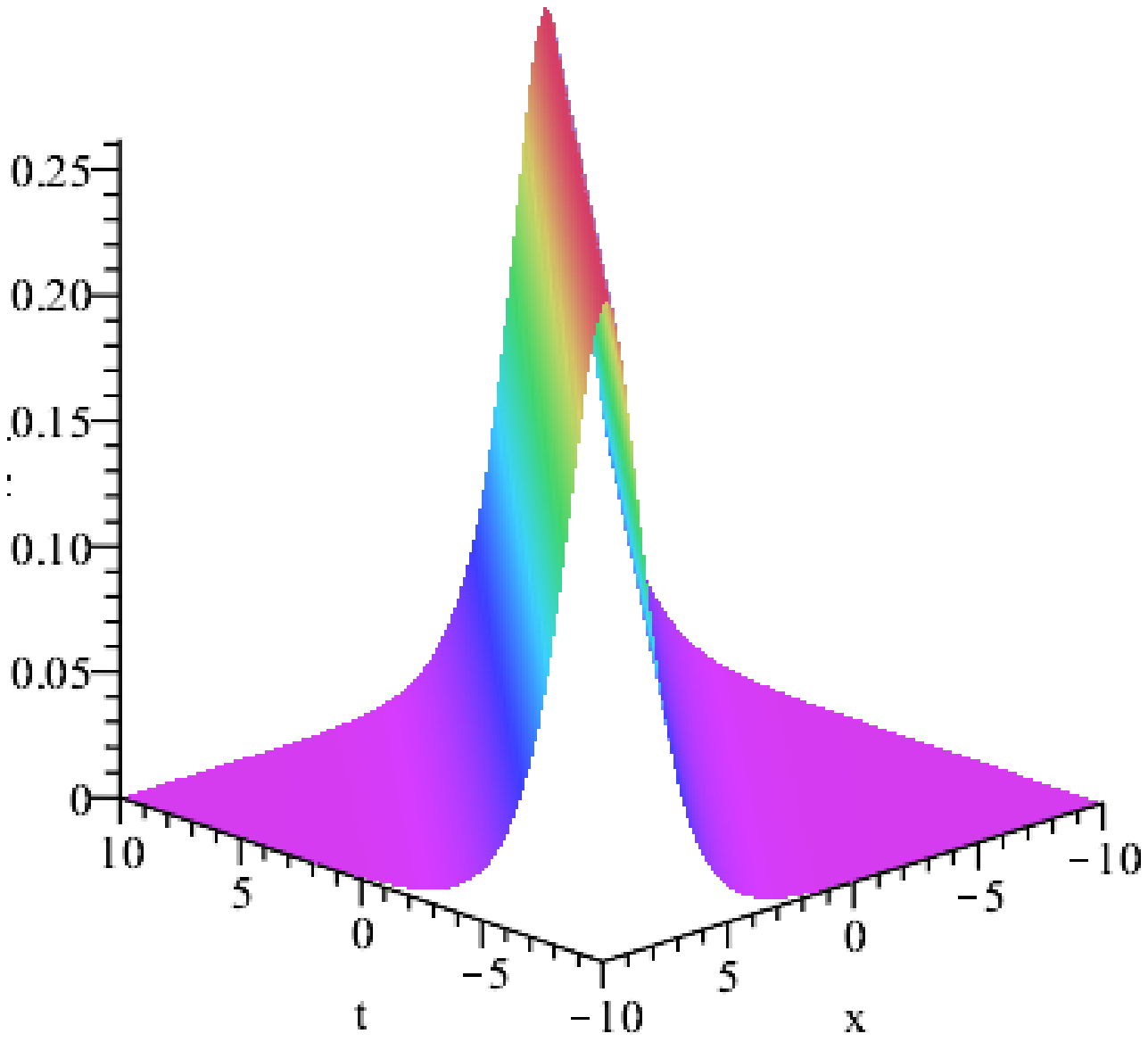}}}
~~~~
{\rotatebox{0}{\includegraphics[width=3.5cm,height=3.5cm,angle=0]{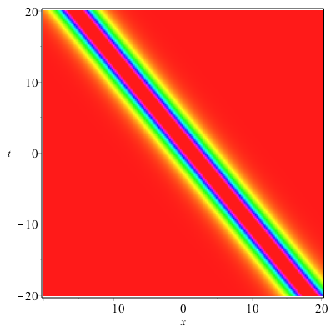}}}
\qquad\quad
{\rotatebox{0}{\includegraphics[width=3.5cm,height=3.5cm,angle=0]{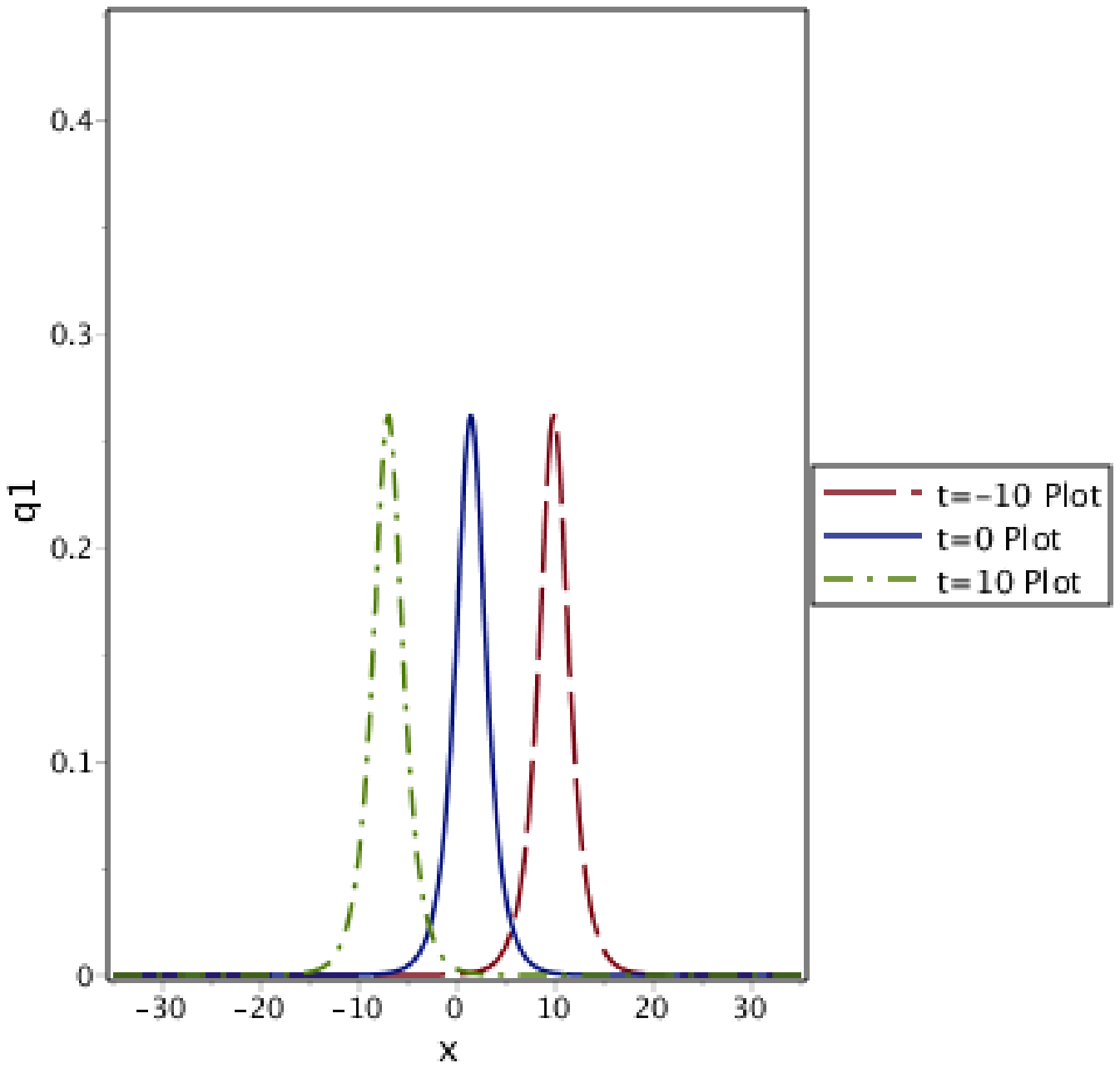}}}

\qquad\quad  $(a1)$
\qquad\qquad\qquad\qquad\qquad $(a2)$ \qquad\qquad\qquad\qquad\qquad$(a3)$\\
\noindent { \small \textbf{Figure 1.}
The single-soliton solutions for $|q_{1}|$ with the parameters selection
$\alpha_{1}=1+2i, \beta_{1}=2+i, \gamma_{1}=2-2i, \xi_{1}=\frac{1}{4}, \eta_{1}=\frac{1}{4}, \epsilon=\frac{1}{12}.$
$\textbf{(a1)}$ three dimensional plot at time $t = 0$  in the $(x, t)$ plane,
$\textbf{(a2)}$ density plot,
$\textbf{(a3)}$ The wave propagation along the $x$-axis with $t = -10$, $t = 0$, $t = 10$.}
\\

\noindent
{\rotatebox{0}{\includegraphics[width=3.5cm,height=3.5cm,angle=0]{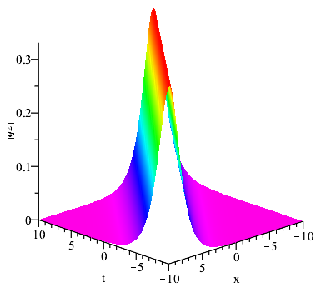}}}
~~~~
{\rotatebox{0}{\includegraphics[width=3.5cm,height=3.5cm,angle=0]{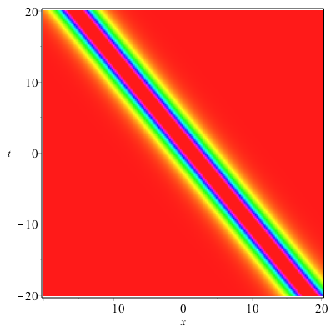}}}
\qquad\quad
{\rotatebox{0}{\includegraphics[width=3.5cm,height=3.5cm,angle=0]{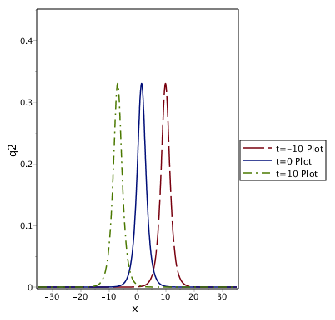}}}

\qquad\quad $(b1)$
\qquad\qquad\qquad\qquad\qquad $(b2)$ \qquad\qquad\qquad\qquad\qquad$(b3)$\\
\noindent { \small \textbf{Figure 2.}
The single-soliton solutions for $|q_{1}|$ with the parameters selection
$\alpha_{1}=1+2i, \beta_{1}=2+i, \gamma_{1}=2-2i, \xi_{1}=\frac{1}{4}, \eta_{1}=\frac{1}{4}, \epsilon=\frac{1}{12}.$
$\textbf{(b1)}$ three dimensional plot at time $t = 0$  in the $(x, t)$ plane,
$\textbf{(b2)}$ density plot,
$\textbf{(b3)}$ The wave propagation along the $x$-axis with $t = -10$, $t = 0$, $t = 10$.}
\\

\noindent
{\rotatebox{0}{\includegraphics[width=3.5cm,height=3.5cm,angle=0]{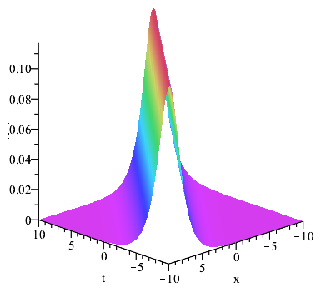}}}
~~~~
{\rotatebox{0}{\includegraphics[width=3.5cm,height=3.5cm,angle=0]{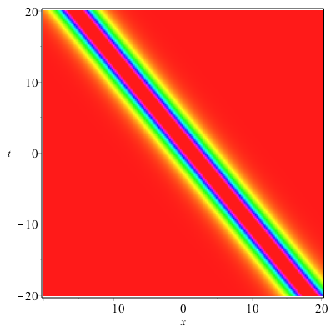}}}
\qquad\quad
{\rotatebox{0}{\includegraphics[width=3.5cm,height=3.5cm,angle=0]{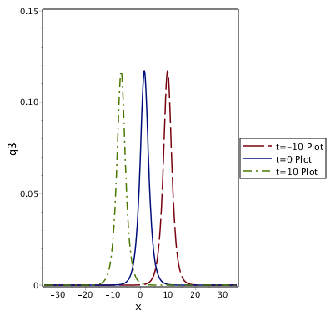}}}

\qquad\quad $(c1)$
\qquad\qquad\qquad\qquad\qquad $(c2)$ \qquad\qquad\qquad\qquad\qquad$(c3)$\\
\noindent { \small \textbf{Figure 3.}
The single-soliton solutions for $|q_{1}|$ with the parameters selection
$\alpha_{1}=1+2i, \beta_{1}=2+i, \gamma_{1}=2-2i, \xi_{1}=\frac{1}{4}, \eta_{1}=\frac{1}{4}, \epsilon=\frac{1}{12}.$
$\textbf{(c1)}$ three dimensional plot at time $t = 0$  in the $(x, t)$ plane,
$\textbf{(c2)}$ density plot,
$\textbf{(c3)}$ The wave propagation along the $x$-axis with $t = -10$, $t = 0$, $t = 10$.}

When taking $N=2$, the two-soliton solution can be expressed by
\begin{align}\label{Q36}
&\left\{ \begin{aligned}
q_{1}=&3ic_{1}\alpha_{1}\beta^{*}_{1}e^{-2\theta_{1}+\theta^{*}_{1}}\left(M^{-1}\right)_{11}
+3ic_{1}\alpha_{1}\beta^{*}_{2}e^{-2\theta_{1}+\theta^{*}_{2}}\left(M^{-1}\right)_{12}\\
&+3ic_{1}\alpha_{2}\beta^{*}_{1}e^{-2\theta_{2}+\theta^{*}_{1}}\left(M^{-1}\right)_{21}
+3ic_{1}\alpha_{2}\beta^{*}_{2}e^{-2\theta_{2}+\theta^{*}_{2}}\left(M^{-1}\right)_{22},\\
q_{2}=&3ic_{1}\alpha_{1}\gamma^{*}_{1}e^{-2\theta_{1}+\theta^{*}_{1}}\left(M^{-1}\right)_{11}
+3ic_{1}\alpha_{1}\gamma^{*}_{2}e^{-2\theta_{1}+\theta^{*}_{2}}\left(M^{-1}\right)_{12}\\
&+3ic_{1}\alpha_{2}\gamma^{*}_{1}e^{-2\theta_{2}+\theta^{*}_{1}}\left(M^{-1}\right)_{21}
+3ic_{1}\alpha_{2}\gamma^{*}_{2}e^{-2\theta_{2}+\theta^{*}_{2}}\left(M^{-1}\right)_{22},\\
q_{3}=&3ic_{1}\alpha_{1}e^{-2\theta_{1}+\theta^{*}_{1}}\left(M^{-1}\right)_{11}
+3ic_{1}\alpha_{1}e^{-2\theta_{1}+\theta^{*}_{2}}\left(M^{-1}\right)_{12}\\
&+3ic_{1}\alpha_{2}e^{-2\theta_{2}+\theta^{*}_{1}}\left(M^{-1}\right)_{21}
+3ic_{1}\alpha_{2}e^{-2\theta_{2}+\theta^{*}_{2}}\left(M^{-1}\right)_{22},
\end{aligned}\right.
\end{align}
where
\begin{align*}
\left\{ \begin{aligned}
m_{11}=\frac{\alpha^{*}_{1}\alpha_{1}e^{-2(\theta^{*}_{1}+\theta_{1})}+\left(\beta^{*}_{1}\beta_{1}+\gamma^{*}_{1}\gamma_{1}+1\right)e^{\theta^{*}_{1}+\theta_{1}}}{\lambda_{1}-\hat{\lambda}_{1}},\\
m_{12}=\frac{\alpha^{*}_{1}\alpha_{2}e^{-2(\theta^{*}_{1}+\theta_{2})}+\left(\beta^{*}_{1}\beta_{2}+\gamma^{*}_{1}\gamma_{2}+1\right)e^{\theta^{*}_{1}+\theta_{2}}}{\lambda_{2}-\hat{\lambda}_{1}},\\
m_{21}=\frac{\alpha^{*}_{2}\alpha_{1}e^{-2(\theta^{*}_{2}+\theta_{1})}+\left(\beta^{*}_{2}\beta_{1}+\gamma^{*}_{2}\gamma_{1}+1\right)e^{\theta^{*}_{2}+\theta_{1}}}{\lambda_{1}-\hat{\lambda}_{2}},\\
m_{22}=\frac{\alpha^{*}_{2}\alpha_{2}e^{-2(\theta^{*}_{2}+\theta_{2})}+\left(\beta^{*}_{2}\beta_{2}+\gamma^{*}_{2}\gamma_{2}+1\right)e^{\theta^{*}_{2}+\theta_{2}}}{\lambda_{2}-\hat{\lambda}_{2}}.
\end{aligned}\right.
\end{align*}

Next, we discuss the case for $N=2$. In Figs. 4-6, the two solitons pass through each other, and the shape changes after collision. In Fig. 4 and Fig. 5, the two soliton peaks change from equilibrium to  large difference, and the soliton peaks change from large difference to almost equal height and in Fig. 6. It is speculated that the soliton energy has a large amount of transfer during collision.\\

\noindent
{\rotatebox{0}{\includegraphics[width=3.5cm,height=3.5cm,angle=0]{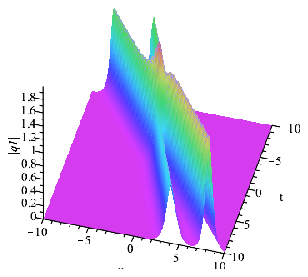}}}
~~~~
{\rotatebox{0}{\includegraphics[width=3.5cm,height=3.5cm,angle=0]{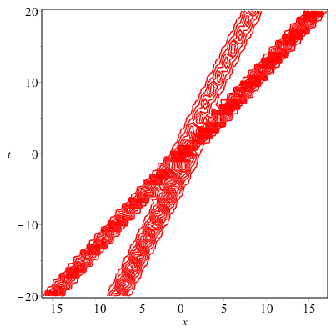}}}
\qquad\quad
{\rotatebox{0}{\includegraphics[width=3.5cm,height=3.5cm,angle=0]{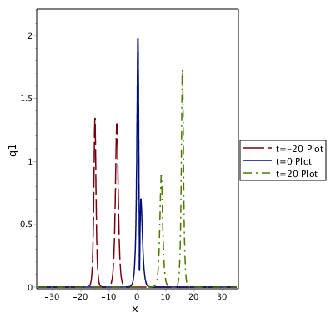}}}

\qquad\quad $(d1)$
\qquad\qquad\qquad\qquad\qquad $(d2)$ \qquad\qquad\qquad\qquad\qquad$(d3)$\\
\noindent { \small \textbf{Figure 4.}
The two-soliton solution for $|q_{1}|$ with the parameters selection
$\alpha_{1}=\frac{\sqrt{3}}{2}+\frac{1}{2}i, \beta_{1}=\frac{\sqrt{3}}{2}+\frac{1}{2}i, \gamma_{1}=\frac{\sqrt{3}}{2}+\frac{1}{2}i,
\alpha_{2}=\frac{\sqrt{3}}{2}-\frac{1}{2}i, \beta_{2}=\frac{\sqrt{3}}{2}-\frac{1}{2}i, \gamma_{2}=\frac{\sqrt{3}}{2}-\frac{1}{2}i,
\xi_{1}=0, \eta_{1}=\frac{3}{4},\xi_{2}=0, \eta_{1}=\frac{3}{4},
\epsilon=\frac{1}{12}.$
$\textbf{(d1)}$ three dimensional plot at time $t = 0$  in the $(x, t)$ plane,
$\textbf{(d2)}$ density plot,
$\textbf{(d3)}$ The wave propagation along the $x$-axis with $t = -20$, $t = 0$, $t = 20$.}\\

\noindent
{\rotatebox{0}{\includegraphics[width=3.5cm,height=3.5cm,angle=0]{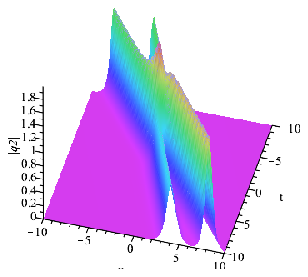}}}
~~~~
{\rotatebox{0}{\includegraphics[width=3.5cm,height=3.5cm,angle=0]{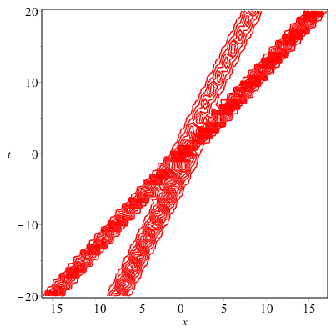}}}
\qquad\quad
{\rotatebox{0}{\includegraphics[width=3.5cm,height=3.5cm,angle=0]{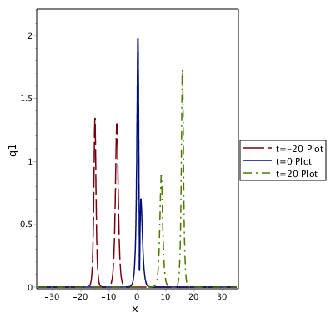}}}

\qquad\quad $(e1)$
\qquad\qquad\qquad\qquad\qquad $(e2)$ \qquad\qquad\qquad\qquad\qquad$(e3)$\\
\noindent { \small \textbf{Figure 5.}
The two-soliton solutions for $|q_{2}|$ with the parameters selection
$\alpha_{1}=\frac{\sqrt{3}}{2}+\frac{1}{2}i, \beta_{1}=\frac{\sqrt{3}}{2}+\frac{1}{2}i, \gamma_{1}=\frac{\sqrt{3}}{2}+\frac{1}{2}i,
\alpha_{2}=\frac{\sqrt{3}}{2}-\frac{1}{2}i, \beta_{2}=\frac{\sqrt{3}}{2}-\frac{1}{2}i, \gamma_{2}=\frac{\sqrt{3}}{2}-\frac{1}{2}i,
\xi_{1}=0, \eta_{1}=\frac{3}{4},\xi_{2}=0, \eta_{1}=\frac{3}{4}, \epsilon=\frac{1}{12}.$
$\textbf{(e1)}$ three dimensional plot at time $t = 0$  in the $(x, t)$ plane,
$\textbf{(e2)}$ density plot,
$\textbf{(e3)}$ The wave propagation along the $x$-axis with $t = -20$, $t = 0$, $t = 20$.}\\

\noindent
{\rotatebox{0}{\includegraphics[width=3.5cm,height=3.5cm,angle=0]{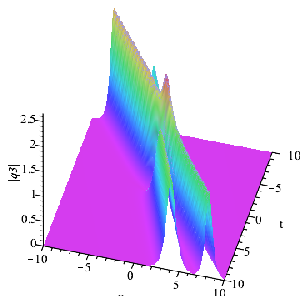}}}
~~~~
{\rotatebox{0}{\includegraphics[width=3.5cm,height=3.5cm,angle=0]{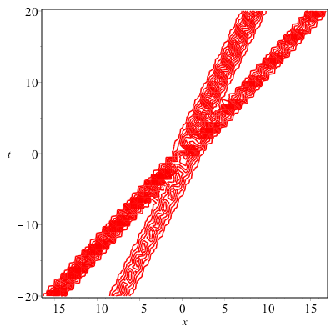}}}
\qquad\quad
{\rotatebox{0}{\includegraphics[width=3.5cm,height=3.5cm,angle=0]{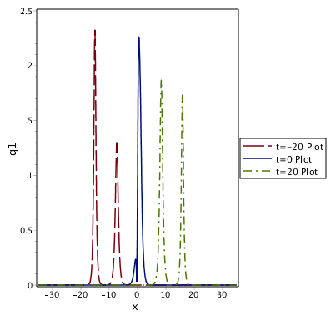}}}

\qquad\quad $(f1)$
\qquad\qquad\qquad\qquad\qquad $(f2)$ \qquad\qquad\qquad\qquad\qquad$(f3)$\\
\noindent { \small \textbf{Figure 6.}
The two-soliton solutions for $|q_{3}|$ with the parameters selection
$\alpha_{1}=\frac{\sqrt{3}}{2}+\frac{1}{2}i, \beta_{1}=\frac{\sqrt{3}}{2}+\frac{1}{2}i, \gamma_{1}=\frac{\sqrt{3}}{2}+\frac{1}{2}i,
\alpha_{2}=\frac{\sqrt{3}}{2}-\frac{1}{2}i, \beta_{2}=\frac{\sqrt{3}}{2}-\frac{1}{2}i, \gamma_{2}=\frac{\sqrt{3}}{2}-\frac{1}{2}i,
\xi_{1}=0, \eta_{1}=\frac{3}{4},\xi_{2}=0, \eta_{1}=\frac{3}{4}, \epsilon=\frac{1}{12}.$
$\textbf{(f1)}$ three dimensional plot at time $t = 0$  in the $(x, t)$ plane,
$\textbf{(f2)}$ density plot,
$\textbf{(f3)}$ The wave propagation along the $x$-axis with $t = -20$, $t = 0$, $t = 20$.}

\section{Conclusions and discussions}

In this work, $N$-soliton solutions of the tcCH equations \eqref{Q1} have been obtained via solving the resulting RH problem based on the inverse scattering theory, which is Eq. \eqref{Q34}. We  have first gotten the analytical property of Jost function via analyzing the lax pair of the tcCH equations. Moreover, the symmetry of the constructed scattering matrix and the time-spatial revolutions of the scattering data has been obtained. According to the above conditions, the RH problem corresponding to the equation has been constructed. Finally, the solution of $N$-soliton solutions with reflection-less  has been calculated.
In addition, some new phenomenon for soliton collision has been presented including localized structures and dynamic behaviors of one- and two- soliton solutions. It is hoped that our results can help enrich the nonlinear dynamics of the $N$-component nonlinear Schr\"{o}dinger type equations.

\section*{Acknowledgements}
This work was supported by the Postgraduate Research and Practice of Educational Reform for Graduate students in CUMT under Grant No. 2019YJSJG046, the Natural Science Foundation of Jiangsu Province under Grant No. BK20181351, the Six Talent Peaks Project in Jiangsu Province under Grant No. JY-059, the Qinglan Project of Jiangsu Province of China, the National Natural Science Foundation of China under Grant No. 11975306, the Fundamental Research Fund for the Central Universities under the Grant Nos. 2019ZDPY07 and 2019QNA35, and the General Financial Grant from the China Postdoctoral Science Foundation under Grant Nos. 2015M570498 and 2017T100413.

\end{document}